\documentclass[letterpaper, 10 pt, conference]{ieeeconf}
\IEEEoverridecommandlockouts    
\usepackage{cite}
\usepackage{amsmath,amssymb,mathrsfs,amsfonts}
\usepackage{algorithmic}
\usepackage{graphicx}
\usepackage{subcaption}
\usepackage{colortbl}
\newtheorem{problem}{Problem}
\newtheorem{definition}{Definition}

\newtheorem{theorem}{Theorem}
\newtheorem{lemma}{Lemma}

\newtheorem{remark}{Remark}

\begin{document}
\title{\LARGE \bf Feedback Interconnected Mean-Field Density Estimation and Control}

\author{Tongjia Zheng$^{1}$, Qing Han$^{2}$ and Hai Lin$^{1}$
\thanks{*This work was supported by the National Science Foundation under Grant No. IIS-1724070, CNS-1830335, IIS-2007949.}
\thanks{$^{1}$Tongia Zheng and Hai Lin are with the Department of Electrical Engineering, University of Notre Dame, Notre Dame, IN 46556, USA. {\tt\small  tzheng1@nd.edu, hlin1@nd.edu.}} 
\thanks{$^{2}$Qing Han is with the Department of Mathematics, University of Notre Dame, Notre Dame, IN 46556, USA. {\tt\small Qing.Han.7@nd.edu.}} 
}

\maketitle

\begin{abstract}
Swarm robotic systems have foreseeable applications in the near future.
Recently, there has been an increasing amount of literature that employs mean-field partial differential equations (PDEs) to model the time-evolution of the probability density of swarm robotic systems and uses density feedback to design stabilizing control laws that act on individuals such that their density converges to a target profile.
However, it remains largely unexplored considering problems of how to estimate the mean-field density, how the density estimation algorithms affect the control performance, and whether the estimation performance in turn depends on the control algorithms.
In this work, we focus on studying the interplay of these algorithms.
Specifically, we propose new density control laws which use the mean-field density and its gradient as feedback, and prove that they are globally input-to-state stable (ISS) with respect to estimation errors.
Then, we design filtering algorithms to estimate the density and its gradient separately, and prove that these estimates are convergent assuming the control laws are known.
Finally, we show that the feedback interconnection of these estimation and control algorithms is still globally ISS, which is attributed to the bilinearity of the PDE system.
An agent-based simulation is included to verify the stability of these algorithms and their feedback interconnection.
\end{abstract}

% \begin{keywords}

% \end{keywords}

\section{Introduction}
Swarm robotic systems (such as drones) have foreseeable applications in the near future.
Compared with small-scale robotic systems, the dramatic increase in the number of involved robots provides numerous advantages such as robustness, efficiency, and flexibility, but also poses significant challenges to their estimation and control problems.

Many methods have been proposed for controlling large-scale systems, such as graph theoretic design \cite{olfati2007consensus} and game theoretic formulation, especially potential games \cite{marden2009cooperative} and mean-field games \cite{huang2006large}.
Our work is also inspired by modelling and control strategies based on mean-field limit approximations. 
However, unlike mean-field games where the mean-field density is used to approximate the collective effect of the swarm, we aim at the direct control of the mean-field density, which results in a PDE control problem.
Mean-field models include Markov chains and PDEs.
The first category partitions the spatial domain to obtain an abstracted Markov chain model and designs the transition probability to stabilize the density \cite{accikmecse2012markov, bandyopadhyay2017probabilistic}, which usually suffers from the state explosion issue.
In PDE-based models, individual robots are modelled by a family of stochastic differential equations and their density evolves according to a PDE.
In this way, the density control problem of a robotic swarm is converted as a regulation problem of the PDE.

Considering density control, early efforts usually adopt an optimal control formulation \cite{elamvazhuthi2015optimal}.
While an optimal control formulation provides more flexibility for the control objective, the solution relies on numerically solving the optimality conditions, which are computationally expensive and essentially open-loop.
Density control is also studied in \cite{chen2015optimal1, ridderhof2019nonlinear} by relating density control problems with the so-called Schr\"odinger Bridge problem.
However, numerically solving the associated Schr\"odinger Bridge problem is also known to suffer from the curse of dimensionality except for the linear case.
In recent years, researchers have sought to design control laws that explicitly use the mean-field density as feedback to form closed-loop control \cite{eren2017velocity, elamvazhuthi2018bilinear, zheng2021transporting}.
This density feedback technique is able to guarantee closed-loop stability and can be efficiently computed since it is given in a closed form.
However, it remains largely unexplored considering problems of how to estimate the mean-field density, how the density estimation algorithm affects the control performance, and whether the estimation performance in turn depends on the density control algorithms.
These problems become more critical as it is observed that most density feedback design more or less depends on the gradient of the mean-field density.
Since gradient is an unbounded operator, any density estimation algorithm that produces accurate density estimates may have arbitrarily large estimation errors for the gradient.
This brings significant concerns to the density estimation problem. 

In this work, we study the interplay of density feedback laws and estimation algorithms.
In \cite{zheng2021transporting}, we have proposed some density feedback laws and obtained preliminary results on their robustness to density estimation errors.
This work extends these feedback laws so that they are less restrictive and have more verifiable robustness properties in the presence of estimation errors. 
We have also reported density filtering algorithms in \cite{zheng2020pde, zheng2021distributedmean} which are designed for large-scale stochastic systems modelled by mean-field PDEs.
This work extends the algorithms to directly estimate the gradient of the density, a quantity required by almost all existing density feedback control in the literature. 
Furthermore, we study the interconnection of these estimation and control algorithms and prove their closed-loop stability.
% Although the design of the feedback velocity field is in a centralized manner, the implementation can be fully distributed in the sense that each robot can independently derive its own low-level controller according to the feedback velocity command. 
% Such a control strategy is suitable for scenarios when a communication center is available, such as UAV-based environment surveillance.

Our contribution includes three aspects.
First, we propose new density feedback laws and show their robustness using the notion of ISS.
Second, we design infinite-dimensional filters to estimate the gradient of the density and study their stability and optimality.
Third, we prove that the feedback interconnection of these estimation and control algorithms is still globally ISS.

The rest of the paper is organized as follow. Section \ref{section:preliminaries} introduces some preliminaries. 
Problem formulation is given in Section \ref{section:problem formulation}. 
Section \ref{section:main results} is our main results in which we propose new density estimation and control laws, and study their interconnected stability. 
Section \ref{section:simulation} presents an agent-based simulation to verify the effectiveness.

\section{Preliminaries}
\label{section:preliminaries}

\subsection{Notations}\label{section:notation}
For a vector $x\in\mathbb{R}^n$, its Euclidean norm is denoted by $\|x\|$.
% Since norms for finite-dimensional vectors are equivalent, we will drop the subscription $p$.
Let $E\subset\mathbb{R}^n$ be a measurable set. 
For $f:E\to\mathbb{R}$, its $L^2$-norm is denoted by $\|f\|_{L^2(E)}:=(\int_{E}|f(x)|^{2}dx)^{1/2}$.
Given $g(x)$ with $0<\inf_{x\in E}g(x)\leq\sup_{x\in E}g(x)<\infty$, the weighted norm $\|f\|_{L^p(E;g)}:=(\int_{E}|f(x)|^{p}g(x)dx)^{1/p}$ is equivalent to $\|f\|_{L^p(E)}$.
We will omit $E$ in the notation when it is clear.
The gradient and Laplacian of a scalar function $f$ are denoted by $\nabla f$ and $\Delta f$, respectively
The divergence of a vector field $F$ is denoted by $\nabla\cdot F$. 
% The differentiation operation of these operators are only taken with respect to the spatial variable $x$ if $f$ and $\boldsymbol{F}$ are also functions of $t$.

% The following lemma is a consequence of Proposition 4.2.7 and Theorem 4.6.3 in \cite{bakry2013analysis}.

\begin{lemma}[Poincar\'e inequality for density functions \cite{bakry2013analysis}] 
\label{lmm:Poincare inequality}
Let $\Omega\in\mathbb{R}^{n}$ be a bounded convex open set with a Lipschitz boundary.
Let $g$ be a continuous density function on $\Omega$ such that $0<c_1\leq\inf g\leq\sup g\leq c_2$ for some constants $c_1$ and $c_2$.
Then, $\exists$ a constant $C>0$ such that for $\forall f\in H^1(\Omega)$,
\begin{equation}\label{eq:poincare inequality}
    \int_\Omega|\nabla f|^2gdx\geq C\int_\Omega\left|f-\int_\Omega fgdx\right|^2gdx.
\end{equation}
\end{lemma}

\subsection{Input-to-state stability}
% Input-to-state stability (ISS) is a stability notion to study nonlinear control systems with external inputs \cite{sontag1989smooth}. 
We introduce ISS for infinite-dimensional systems \cite{dashkovskiy2013input}.
Define the following classes of comparison functions:
\begin{align*}
    \mathscr{P} &:=\{\gamma:\mathbb{R}_+\to\mathbb{R}_+|\gamma\text{ is continuous, }\gamma(0)=0,\\
    &\qquad\text{ and }\gamma(r)>0\text{ for }r>0\}\\
    \mathscr{K} &:=\{\gamma\in\mathscr{P}\mid\gamma\text{ is strictly increasing}\} \\
    \mathscr{K}_\infty &:=\{\gamma\in\mathscr{K}\mid\gamma\text{ is unbounded}\}\\
    \mathscr{L} &:=\{\gamma:\mathbb{R}_+\to\mathbb{R}_+\mid\gamma\text{ is continuous and strictly}\\
    &\quad\quad\text{decreasing with }\lim_{t\to\infty}\gamma(t)=0\}\\
    \mathscr{KL} &:=\{\beta:\mathbb{R}_+\times\mathbb{R}_+\to\mathbb{R}_+\mid\beta(\cdot,t)\in\mathscr{K},\forall t\geq0,\\
    &\qquad\beta(r,\cdot)\in\mathscr{L},\forall r>0\}.
\end{align*}

Let $\left(X,\|\cdot\|_X\right)$ and $\left(U,\|\cdot\|_{U}\right)$ be the state and input space, endowed with norms $\|\cdot\|_X$ and $\|\cdot\|_{U}$, respectively.
Denote $U_c=PC(\mathbb{R}_+;U)$, the space of piecewise right-continuous functions from $\mathbb{R}_+$ to $U$, equipped with the sup-norm.
Consider a control system $\Sigma=(X,U_c,\phi)$ where $\phi: \mathbb{R}_+\times X\times U_c\to X$ is a transition map.
Let $x(t)=\phi(t,x_0,u)$.
%the system state at time $t\in\mathbb{R}_+$, if its state at time $s\in\mathbb{R}_+$ was $x \in X$ and the input $u \in U_c$ was applied.

\begin{definition} \label{dfn:(L)ISS}
$\Sigma$ is called \textit{input-to-state stable (ISS)}, if $\exists\beta\in\mathcal{KL},\gamma\in\mathcal{K}$, such that
\begin{equation*}\label{eq:(L)ISS}
    \|x(t)\|_X\leq\beta(\|x_0\|_X, t)+\gamma\Big(\sup_{0\leq s\leq t}\|u(s)\|_{U}\Big),
\end{equation*}
$\forall x_0\in X,\forall u\in U_c$ and $\forall t\geq0$.
% It is called \textit{input-to-state stable (ISS)}, if $\rho_x=\infty$ and $\rho_u=\infty$.
\end{definition} 

% To verify the ISS property, Lyapunov functions can be exploited.

\begin{definition}\label{dfn:(L)ISS-Lyapunov function}
A continuous function $V:\mathbb{R}_+\times X\to\mathbb{R}_+$ is called an \textit{ISS-Lyapunov function} for $\Sigma$, if $\exists\psi_{1},\psi_{2}\in\mathcal{K}_\infty,\chi\in\mathcal{K}$, and $W\in\mathcal{P}$, such that:
\begin{itemize}
    \item[(i)] $\psi_1(\|x\|_X)\leq V(t,x)\leq\psi_2(\|x\|_X), ~\forall x\in X$
    \item[(ii)] $\forall x\in X,\forall u\in U_c$ with $u(0)=\xi\in U$ it holds:
    \begin{equation*}
        \|x\|_X\geq\chi(\|\xi\|_U) \Rightarrow \dot{V}(t,x)\leq-W(\|x\|_X).
    \end{equation*}
\end{itemize}
% If $D=X, \rho_x=\infty$ and $\rho_u=\infty,$ then the function $V$ is called an \textit{ISS-Lyapunov function}.
\end{definition}

\begin{theorem}\label{thm:(L)ISS-Lyapunov function}
% Let $\Sigma=(X,U_c,\phi)$ be a control system, and $x\equiv0$ be its equilibrium point.
If $\Sigma$ admits an ISS-Lyapunov function, then it is ISS.
\end{theorem}

ISS is useful for studying the stability of cascade systems.
Consider two systems $\Sigma_i=(X_i,U_{ci},\phi_i),i=1,2$, where $U_{ci}=PC(\mathbb{R}_+;U_i)$ and $X_1\subset U_2$.
We say they form a cascade connection if $u_2(t)=\phi_1(t,t_0,\phi_{01},u_1)$.

\begin{theorem}\label{thm:(L)ISS cascade}
\cite{khalil2002nonlinear} The cascade connection of two ISS systems is ISS.
% If one of them is LISS, then the cascade connection is LISS.
\end{theorem}

% \begin{definition} \label{dfn:(L)ISS}
% $\Sigma$ is called \textit{input-to-state stable (ISS)}, if $\exists\beta\in\mathscr{KL},\gamma\in\mathscr{K}$, such that
% \begin{equation*}\label{eq:(L)ISS}
%     \|x(t)\|_X\leq\beta(\|x_0\|_X, t)+\gamma\Big(\sup_{0\leq s\leq t}\|u(s)\|_{U}\Big),
% \end{equation*}
% $\forall x_0\in X,\forall u\in U_c$ and $\forall t\geq0$.
% % It is called \textit{input-to-state stable (ISS)}, if $\rho_x=\infty$ and $\rho_u=\infty$.
% \end{definition} 

% % To verify the ISS property, Lyapunov functions can be exploited.

% \begin{definition}\label{dfn:(L)ISS-Lyapunov function}
% A continuous function $V:\mathbb{R}_+\times X\to\mathbb{R}_+$ is called an \textit{ISS-Lyapunov function} for $\Sigma$, if $\exists\psi_{1},\psi_{2}\in\mathscr{K}_\infty,\chi\in\mathscr{K}$, and $W\in\mathscr{P}$, such that:
% \begin{itemize}
%     \item[(i)] $\psi_1(\|x\|_X)\leq V(t,x)\leq\psi_2(\|x\|_X), ~\forall x\in X$
%     \item[(ii)] $\forall x\in X,\forall u\in U_c$ with $u(0)=\xi\in U$ it holds:
%     \begin{equation*}
%         \|x\|_X\geq\chi(\|\xi\|_U) \Rightarrow \dot{V}(t,x)\leq-W(\|x\|_X).
%     \end{equation*}
% \end{itemize}
% % If $D=X, \rho_x=\infty$ and $\rho_u=\infty,$ then the function $V$ is called an \textit{ISS-Lyapunov function}.
% \end{definition}

% \begin{theorem}\label{thm:ISS-Lyapunov function}
% % Let $\Sigma=(X,U_c,\phi)$ be a control system, and $x\equiv0$ be its equilibrium point.
% If $\Sigma$ admits an ISS-Lyapunov function, then it is ISS.
% \end{theorem}

\subsection{Infinite-dimensional Kalman filters}
We introduce the infinite-dimensional Kalman filters presented in \cite{curtain1975infinite}.
Let $\mathcal{H},\mathcal{K}$ be real Hilbert spaces. 
Consider the following infinite-dimensional linear system:
\begin{align*}
&du(t)=\mathcal{A}(t)u(t)dt+\mathcal{B}(t)dw(t), ~u(0)=u_0 \\
&dz(t)=\mathcal{C}(t)u(t)dt+\mathcal{F}(t)dv(t), ~z(0)=0,
\end{align*}
where $\mathcal{A}(t)$ is a linear closed operator on $\mathcal{H}$, $\mathcal{B}(\cdot)\in L^\infty([0,T];\mathcal{L}(\mathcal{H}))$, $\mathcal{C}(\cdot)\in L^\infty([0,T];\mathcal{L}(\mathcal{H},\mathcal{K}))$, and $\mathcal{F}(\cdot),\mathcal{F}(\cdot)^{-1}\in L^\infty([0,T];\mathcal{L}(\mathcal{K}))$.
% $u_0$ is an $\mathcal{H}$-valued random variable independent of $w(t)$ and $v(t)$ and has zero expectation and covariance operator $\mathcal{P}_0$. 
$w(t)$ and $v(t)$ are independent Wiener processes on $\mathcal{H}$ and $\mathcal{K}$ with covariance operators $\mathcal{W}$ and $\mathcal{V}$, respectively. 
% We assume that the observation space, $\mathcal{K},$ is finite-dimensional. 
% Then from the results (1.9), (2.1) has the unique mild solution
The infinite-dimensional Kalman filter is given by:
\begin{equation*}
    d\hat{u}(t)=\mathcal{A}(t)\hat{u}dt+\mathcal{K}(t)(dz(t)-\mathcal{C}(t)u(t)dt),~\hat{u}(0)=0
\end{equation*}
where $\mathcal{K}(t)=\mathcal{P}(t)\mathcal{C}^{*}(t)\left(\mathcal{F}(t)\mathcal{V}\mathcal{F}^{*}(t)\right)^{-1}$ is the Kalman gain, and $\mathcal{P}(t)$ is the solution of the Riccati equation:
\begin{align*}
    \frac{d\mathcal{P}(t)}{dt}&=\mathcal{A}(t)\mathcal{P}(t)+\mathcal{P}(t)\mathcal{A}^{*}(t)+\mathcal{B}(t)\mathcal{W}\mathcal{B}^{*}(t)\\
    &\quad-\mathcal{P}(t)\mathcal{C}^{*}(t)\left(\mathcal{F}(t)\mathcal{V}\mathcal{F}^{*}(t)\right)^{-1}\mathcal{C}(t)\mathcal{P}(t).
\end{align*}

\section{Problem formulation} \label{section:problem formulation}
This work studies the density control problem of robotic swarms.
% Specifically, we want to design velocity commands for individual robots such that their density evolves to a target profile. 
Consider $N$ robots in a bounded convex domain $\Omega\subset\mathbb{R}^n$, which are assumed to be homogeneous and satisfy:
\begin{equation} \label{eq:Langevin equation}
    dX_t^i=v(X_t^i,t)dt+\sqrt{2\sigma(t)}dB_t^i, \quad i = 1,\dots,N,
\end{equation}
where $X_t^i\in\Omega$ represents the position of the $i$-th robot, $v=(v_1,\dots,v_n)\in\mathbb{R}^n$ is the velocity field to be designed, $\{B_t^i\}$ are standard Wiener processes assumed to be independent across the robots, and $\sqrt{2\sigma(t)}\in\mathbb{R}$ is the standard deviation.

In the mean-field limit as $N\to\infty$, the collective behavior of the robots can be captured by their probability density
\begin{align*}
    p(x,t)\approx\frac{1}{N}\sum_{i=1}^N\delta_{X_t^i}
\end{align*}
with $\delta_{x}$ being the Dirac distribution, and this density satisfies a Fokker-Planck equation given by:
\begin{align} \label{eq:FP equation}
\begin{split}
    \partial_t p =-\nabla\cdot({v} p) + \Delta(\sigma p) &\quad\text{in}\quad \Omega\times(0,T), \\
    p=p_0 &\quad\text{on}\quad \Omega\times\{0\},\\
    \boldsymbol{n} \cdot(\nabla(\sigma p)-{v}p)=0 &\quad\text{on}\quad \partial\Omega\times(0,T),
\end{split}
\end{align}
where $\boldsymbol{n}$ is the unit inner normal to the boundary $\partial\Omega$, and $p_0(x)$ is the initial density. 
The last equation is the \textit{reflecting boundary condition} to confine the robots within $\Omega$. 

\begin{remark}
% Note that \eqref{eq:Langevin equation} and \eqref{eq:FP equation} share the same coefficients.
% Hence, the velocity field $v$ we design based on \eqref{eq:FP equation} can be easily implemented on individual \eqref{eq:Langevin equation}. 
This work focuses on the interconnected stability of density estimation and control.
For clarity, the robots are assumed to be first-order integrators in \eqref{eq:Langevin equation}.
Density control for heterogeneous higher-order nonlinear systems is studied in a separate work \cite{zheng2021backstepping}.
The interconnected stability results to be presented later can be generalized to these more general systems by combining the stability results in \cite{zheng2021backstepping}.
% Then, individual robots need to derive their own low-level controller to track the reference velocity command, which can however be done in a distributed way.
% The velocity tracking problem has been widely studied in literature especially for mobile robots, and hence is not studied in this work.
\end{remark}

The problems studied in this work are stated as follow.

\begin{problem}[Density control]
Consider \eqref{eq:FP equation}.
Given a target density $p_*(x)$, design  $v$ such that $p(x,t)\to p_*(x)$.
\end{problem}

\begin{problem}[Density estimation]
Consider \eqref{eq:FP equation}.
Given the robots' states $\{X_t^i\}_{i=1}^N$, estimate $p$ and its gradient.
\end{problem}

\begin{problem}[Feedback interconnection]
Consider \eqref{eq:FP equation}.
Given a target density $p_*(x)$, let $v$ be designed as a feedback function of certain density estimates that are computed based on \eqref{eq:FP equation}.
Prove that $p(x,t)\to p_*(x)$.
\end{problem}

\section{Main results}
\label{section:main results}
\subsection{Modified density feedback laws}

Given a smooth target density $p_*(x)$, bounded from above and below by positive constants, we design:
\begin{equation}\label{eq:modified density feedback law 1}
    v(x,t)=-\alpha(x,t)\nabla\frac{p(x,t)}{p_*(x)}+\frac{\sigma(t)\nabla p_*(x)}{p_*(x)}
\end{equation}
where $\alpha\geq0$ is a design parameter for individuals to adjust their velocity magnitude.
In the collective level, \eqref{eq:modified density feedback law 1} is called density feedback because of its explicit dependence on the mean-field density $p$.
In the individual level, \eqref{eq:modified density feedback law 1} is essentially nonlinear state feedback and the velocity input for the $i$-th robot is simply given by $v(X_t^i,t)$. 

\begin{remark}
Compared with the control laws proposed in \cite{eren2017velocity, elamvazhuthi2018bilinear, zheng2021transporting}, the remarkable difference of \eqref{eq:modified density feedback law 1} is that $p$ does not appear in any denominator, which provides several advantages.
First, it relaxes the requirement for $p$ to be strictly positive and avoids producing large velocity when $p$ is close to 0.
% In fact, it doesn't even require $p$ to be positive, which is very useful because there exist density estimation algorithms that can generate negative estimates.
Second, it will enable us to obtain ISS results with respect to estimation errors in $L^2$ norm.
(Such results are difficult to obtain for those in \cite{eren2017velocity, elamvazhuthi2018bilinear, zheng2021transporting}.)
The significance of this property will become apparent when we study the interconnected stability of control and estimation algorithms.
\end{remark}

% In this section, we focus on the stability and robustness of \eqref{eq:modified density feedback law 1}.
Define $\Phi=p-p_*$ as the convergence error.
% We require the following assumption.

% \begin{assumption}\label{assump:p*}
% Assume $p_*=e^{-V(x)}$ where $D^{2}V\geq\lambda I_{n}>0$.
% In other words, $V$ is convex on $\mathbb{R}^n$, and the measure corresponding to $p_*$ is a log-concave measure. To restrict to a convex $\Omega$, set $V$ to be infinity outside $\Omega$. Does this give flexibility to the choice of $p_*$?
% \end{assumption}

% \begin{remark}
% The reason for this assumption is that we need the weighted Poinc\'are inequality \eqref{eq:poincare inequality} to prove exponential stability.
% However, we believe that this assumption is not necessary since the weighted Poinc\'are inequality has been shown to hold for many circumstances \cite{pechstein2013weighted}.
% We will study how to drop this assumption in future work.
% \end{remark}

\begin{theorem}
\label{thm:exponential stability}
Consider system \eqref{eq:FP equation}.
Let $v$ be given by \eqref{eq:modified density feedback law 1}.
Then $\|\Phi(t)\|_{L^2}$ converges to 0 exponentially.
\end{theorem}

% relationship with KL divergence, may be related to L^2 after exponential variable transform

\begin{proof}
Substituting \eqref{eq:modified density feedback law 1} into \eqref{eq:FP equation}, we obtain:
\begin{align*}
\begin{split}
    \partial_tp &=\nabla\cdot\left(\alpha p\nabla\frac{p}{p_*}\right) - \nabla\cdot\left(\frac{\sigma p}{p_*}\nabla p_*\right)+\Delta(\sigma p) \\
    &=\nabla\cdot\left(\alpha p\nabla\frac{p}{p_*}\right) - \nabla\cdot\left[\frac{\sigma p}{p_*}\nabla p_*-\nabla(\sigma p)\right] \\
    &=\nabla\cdot\left(\alpha p\nabla\frac{p}{p_*}\right) + \nabla \cdot\left(\sigma p_*\nabla\frac{p}{p_*}\right).
\end{split}
\end{align*}
Consider a Lyapunov function $V(t)=\frac{1}{2}\|\Phi\|_{L^2(1/p_*)}^2$.
% =\frac{1}{2}\int \frac{\left(p-p_*\right)^{2}}{p_*} d x=\frac{1}{2}\int\left(\frac{p}{p_*}-1\right)^{2} p_* d x
By the divergence theorem and the boundary condition, we have
\begin{align*}
\begin{split}
    \frac{dV}{dt}&=\int_\Omega\frac{p-p_*}{p_*} \partial_tpdx\\
    &=\int_\Omega\frac{p-p_*}{p_*}\left[\nabla\cdot\left(\alpha p\nabla\frac{p}{p_*}\right) + \nabla\cdot\left(\sigma p_*\nabla\frac{p}{p_*}\right)\right]dx \\
    &=\int_\Omega-\alpha p\left(\nabla\frac{p}{p_*}\right)^{2} - \sigma p_*\left(\nabla\frac{p}{p_*}\right)^{2}dx\\
    &\leq-(\alpha_1(t)+\alpha_2(t))\left\|\nabla\frac{p}{p_*}\right\|_{L^2}^2\\
    &=-(\alpha_1(t)+\alpha_2(t))\left\|\nabla\frac{p-p_*}{p_*}\right\|_{L^2}^2.
\end{split}
\end{align*}
By the Poinc\'are inequality \eqref{eq:poincare inequality} (where we set $f=\frac{p-p_*}{p_*}$ and $g=p_*$) and the fact that $\int_\Omega(p-p_*)dx=0$, we have
\begin{align*}
\begin{split}
    \frac{dV}{dt}&\leq-(\alpha_1(t)+\alpha_2(t))C^2\left\|\frac{p-p_*}{p_*}\right\|_{L^2}^2\\
    &=-(\alpha_1(t)+\alpha_2(t))C^2\left\|\Phi\right\|_{L^2(1/p_*^2)}^2
\end{split}
\end{align*}
where $\alpha_1(t),\alpha_2(t)\in\mathbb{R}$ satisfy $0\leq\min_x\alpha p\leq\alpha_1(t)\leq\max_x\alpha p$ and $0<\min_x\sigma p_*\leq\alpha_2(t)\leq\max_x\sigma p_*$, and $C>0$ is the Poinc\'are constant.
By the strong maximum principle \cite{lieberman1996second}, there exist $t_1,a>0$ such that $p\geq a$ for $t\geq t_1$.
Hence, $\alpha_1=0$ if and only if $\alpha=0$ for $t\geq t_1$.
\end{proof}

% {\color{red}That is a mistake.
% I intended to use the following Poincare inequality for probability measures.
% $$
% \int_\Omega|\nabla g|^2p_*dx\geq\lambda\int_\Omega\left|g-\int_\Omega gp_*dx\right|^2p_*dx
% $$
% where $p_*$ is a probability measure satisfying $p_*(x)=e^{-V(x)}$ and $V$ needs to satisfy certain conditions.
% I forgot that $p_*$ needs to this condition, which is restrictive.
% I think at least the steps before ``(by the Poincare inequality)'' are correct, for which we only obtain convergence of $\|p-p_*\|_{L^2}\to0$ instead of exponential convergence.
% But I still need something like $\left\|\nabla\frac{p}{p_*}\right\|_{L^2}\geq C\left\|p-p_*\right\|_{L^2}$ in the proof of Theorem 3 to complete the proof of ISS.
% I don't know how to fix it at this point. 
% Do you have any thoughts?  
% It seems that the Poincare inequality is critical.
% }

\begin{remark}
Note that we allow $\alpha=0$ in \eqref{eq:modified density feedback law 1}.
In this case, \eqref{eq:modified density feedback law 1} becomes $v=\frac{\sigma\nabla p_*}{p_*}$, which is a well-known law to drive stochastic particles towards a target distribution in physics \cite{markowich2000trend}.
% In other words, the open-loop system is exponentially stable.
However, the convergence speed will be very slow since $\sigma$ is small in general.
We add the feedback term $-\alpha\nabla\frac{p}{p_*}$ to provide extra and locally adjustable convergence speed.
The accelerated convergence is reflected by $\alpha_1$ in the proof.
\end{remark}

The mean-field density $p$ cannot be measured directly.
We will study how to estimate $p$ in next section.
We first establish some robustness results regardless of what estimation algorithm to use.
It is useful to rewrite \eqref{eq:modified density feedback law 1} as:
\begin{equation}\label{eq:modified density feedback law 2}
    v=-\alpha\frac{p_*\nabla p-p\nabla p_*}{p_*^2}+\frac{\sigma(t)\nabla p_*}{p_*}.
\end{equation}
In general, we need to estimate $p$ and $\nabla p$ separately because $\nabla$ is an unbounded operator, i.e., any algorithm that produces accurate estimates of $p$ may have arbitrarily large estimation errors for $\nabla p$.
% Thus, we need to design additional algorithms to estimate $\nabla p$.
Let $\hat{p}$ and $\widehat{\nabla p}$ be the estimates of $p$ and $\nabla p$.
Based on \eqref{eq:modified density feedback law 2}, the control law using estimates is given by:
\begin{equation}\label{eq:modified density feedback law using estimates}
    v=-\alpha\frac{p_*\widehat{\nabla p}-\hat{p}\nabla p_*}{p_*^2}+\frac{\sigma(t)\nabla p_*}{p_*}.
\end{equation}
Define $\epsilon=\hat{p}-p$ and $\epsilon_g=\widehat{\nabla p}-\nabla p$ as the estimation errors.
Substituting \eqref{eq:modified density feedback law using estimates} into \eqref{eq:FP equation}, we obtain the closed-loop system:
\begin{equation} \label{eq:tracking error equation}
    \partial_t\Phi = \partial_tp =\nabla\cdot\left(\alpha p\frac{p_*\widehat{\nabla p}-\hat{p}\nabla p_*}{p_*^2}+\sigma p_*\nabla\frac{p}{p_*}\right).
\end{equation}

We have the following robustness result.
 
\begin{theorem}\label{thm:ISS}
Consider \eqref{eq:FP equation}.
Let $v$ be given by \eqref{eq:modified density feedback law using estimates}.
If $\alpha>0$, then $\|\Phi(t)\|_{L^2}$ is ISS with respect to $\|\epsilon(t)\|_{L^2}$ and $\|\epsilon_g(t)\|_{L^2}$.
\end{theorem}

\begin{proof}
Consider $V(t)=\frac{1}{2}\|\Phi\|_{L^2(1/p_*)}^2$.
By the divergence theorem and the boundary condition, we have
\begin{align*}
\begin{split}
    \frac{dV}{dt}&=\int_\Omega\frac{p-p_*}{p_*}\left[\nabla\cdot\left(\alpha p\frac{p_*\widehat{\nabla p}-\hat{p}\nabla p_*}{p_*^2}+\sigma p_*\nabla\frac{p}{p_*}\right)\right]dx \\
    % &=\int\frac{p-p_*}{p_*}\left[\nabla\cdot\left(\alpha p\frac{p_*\nabla p-p\nabla p_*+p_*\epsilon_g-\epsilon\nabla p_*}{p_*^2}\right) + \nabla \cdot\left(\sigma p_*\nabla\frac{p}{p_*}\right)\right]dx \\
    &=\int_\Omega-\alpha p\nabla\frac{p}{p_*}\cdot\left(\frac{p_*\nabla p-p\nabla p_*+p_*\epsilon_g-\epsilon\nabla p_*}{p_*^2}\right)\\
    &\quad-\sigma p_*\left(\nabla\frac{p}{p_*}\right)^{2}dx\\
    &=\int_\Omega-\alpha p\nabla\frac{p}{p_*}\cdot\left(\nabla\frac{p}{p_*}+\frac{1}{p_*}\epsilon_g-\frac{\nabla p_*}{p_*^2}\epsilon\right)\\
    &\quad-\sigma p_*\left(\nabla\frac{p}{p_*}\right)^{2}dx\\
    &\leq-(\alpha_1(t)+\alpha_2(t))\left\|\nabla\frac{p}{p_*}\right\|_{L^2}^2\\
    &\quad+\alpha_3(t)\left\|\nabla\frac{p}{p_*}\right\|_{L^2}\left\|\frac{1}{p_*}\epsilon_g-\frac{\nabla p_*}{p_*^2}\epsilon\right\|_{L^2}\\
    % &\leq-(\alpha_1(t)+\alpha_2(t))C^2\left\|\frac{p-p_*}{p_*}\right\|_{L^2}^2
    % &\quad+\alpha_3(t)C\left\|\frac{p-p_*}{p_*}\right\|_{L^2}\left\|\frac{1}{p_*}\epsilon_g-\frac{\nabla p_*}{p_*^2}\epsilon\right\|_{L^2}
\end{split}
\end{align*}
where $\alpha_1$ and $\alpha_2$ are defined in the proof of Theorem \ref{thm:exponential stability}, and $\alpha_3(t)\in\mathbb{R}$ satisfies $0\leq\min_x\alpha p\leq\alpha_3(t)\leq\max_x\alpha p$.
Using a constant $\theta\in(0,1)$ to split the first term and applying the Poinc\'are inequality \eqref{eq:poincare inequality}, we have
\begin{align*}
\begin{split}
    \frac{dV}{dt}&\leq-(\alpha_1(t)+\alpha_2(t))(1-\theta)C^2\left\|\frac{p-p_*}{p_*}\right\|_{L^2}^2\\
    &\quad-(\alpha_1(t)+\alpha_2(t))\theta C\left\|\nabla\frac{p}{p_*}\right\|_{L^2}\left\|\frac{p-p_*}{p_*}\right\|_{L^2}\\
    &\quad +\alpha_3(t)\big(\|\epsilon_g\|_{L^2(\frac{1}{p_*})}+\|\nabla p_*\|_{L^\infty}\|\epsilon\|_{L^2(\frac{1}{p_*^2})}\|\big)\left\|\nabla\frac{p}{p_*}\right\|_{L^2}\\
    &\leq-(\alpha_1(t)+\alpha_2(t))(1-\theta)C^2\|\Phi\|_{L^2(1/p_*^2)}^2,
\end{split}
\end{align*}
if $\|\Phi\|_{L^2(\frac{1}{p_*^2})}\geq\alpha_{\sup}(\|\epsilon_g\|_{L^2(\frac{1}{p_*})}+\|\nabla p_*\|_{L^\infty}\|\epsilon\|_{L^2(\frac{1}{p_*^2})}\|)$ where $\alpha_{\sup}:=\sup_t\frac{\alpha_3(t)}{(\alpha_1(t)+\alpha_2(t))\theta C}$.
The ISS property then follows from Theorem \ref{thm:(L)ISS-Lyapunov function}.
\end{proof}

% By using feedback in the control law \eqref{eq:modified density feedback law 1}, we are able to provide extra design flexibility for the convergence speed and individual velocity magnitude, which however brings potential robustness issue in terms of estimation errors.
This theorem ensures that $\|\Phi(t)\|_{L^2}$ is always bounded by a function of $\|\epsilon(t)\|_{L^2}$ and $\|\epsilon_g(t)\|_{L^2}$, and will approach 0 exponentially if both $\hat{p}$ and $\widehat{\nabla p}$ are accurate.

\subsection{Density and gradient estimation}

Now we design algorithms to estimate $p$ and $\nabla p$ separately.
The corresponding algorithms will be referred to as density filters and gradient filters, respectively.
The former has been studied in \cite{zheng2020pde,zheng2021distributedmean} and is included for completeness.
In the estimation problem, the system \eqref{eq:FP equation} is assumed known (including $v$).
We use kernel density estimation (KDE) to construct noisy measurements of $p$ and $\nabla p$ and utilize two important properties: the dynamics governing $p$ and $\nabla p$ are linear; their measurement noises are approximately Gaussian.
Then we use infinite-dimensional Kalman filters to design two filters to estimate them separately.

% First, it is important to point out that if $v$ is given as a function of $p$, the distribution of $\{X_t^i\}$, then the family of stochastic processes \eqref{eq:Langevin equation} are dependent of each other in general.
% Such stochastic processes are known as the McKean-Vlasov processes.
First, the processes $\{X_t^i\}_{i=1}^{N}$ become asymptotically independent as $N\to\infty$ \cite{huang2006large}.
% Essentially this is because their interaction becomes weaker and weaker in the sense that the contribution of a given particle $i$ is of order $1/N$.
% Hence, as $N\to\infty$, we expect the interaction to vanish, in the sense that a given particle $i$ does not appear in the measure term anymore. 
% If particles do not affect the measure flow, then the particles should be i.i.d., as they have the same coefficients $b$ and $\sigma$ and are driven by independent Brownian motions.
Hence, for a large $N$, $\{X_t^i\}_{i=1}^{N}$ can be approximately treated as independent samples of $p(x,t)$. 
We then use KDE to construct priori estimates which are used as noisy measurements of $p$ and $\nabla p$.

We need to derive an evolution equation for $\nabla p$.
By applying the gradient operator on both sides of \eqref{eq:FP equation}, we have:
\begin{align*}
\begin{split}
    \partial_t(\nabla p)=\nabla(\partial_tp)&=\nabla[-\nabla\cdot(vp) + \Delta(\sigma p)]\\
    &=\nabla[-\nabla\cdot(v\mathscr{I}(\nabla p)) + \Delta(\sigma\mathscr{I}(\nabla p))]
\end{split}
\end{align*}
where $\mathscr{I}$ is an integration operator defined as follows.
For a given vector field $F:\Omega\to\mathbb{R}^n$, define $\mathscr{I}(F)=f$, where $f:\Omega\to\mathbb{R}$ is uniquely determined by the following relations:
\begin{equation}
     \nabla f=F\text{ and }\boldsymbol{n}\cdot(\nabla(\sigma f)-{v}f)=0\text{ on } \partial\Omega.
\end{equation}
For simplicity, denote $q=\nabla p$, i.e., $q:\Omega\times(0,T)\to\mathbb{R}^n$.
We obtain a partial-integro-differential equation for $q$:
\begin{align}\label{eq:gradient equation}
\begin{split}
    \partial_tq=\nabla[-\nabla\cdot(v\mathscr{I}(q)) + \Delta(\sigma\mathscr{I}(q))]
\end{split}
\end{align}
which is a linear equation.
We will rewrite \eqref{eq:FP equation} and \eqref{eq:gradient equation} as evolution equations in $\mathcal{H}_1=L^2(\Omega)$ and $\mathcal{H}_2=L^2(\Omega)^n:=\{(F_1,\dots,F_n):\Omega\to\mathbb{R}^n\mid F_i\in L^2(\Omega)\}$, respectively.
% , and use KDE to construct noisy measurements for their states.
Specifically, define the following linear operators:
\begin{align*}
    &\mathcal{A}(t)p:=-\nabla\cdot[v(t)p] + \Delta[\sigma(t)p],\\
    &\mathcal{A}_g(t)q:=\nabla[-\nabla\cdot(v(t)\mathscr{I}(q)) + \Delta(\sigma(t)\mathscr{I}(q))].
\end{align*}
For any $t$, we use KDE and the samples $\{X_t^i\}_{i=1}^N$ to construct a priori estimate $p_{\text{KDE}}(x,t)$ \cite{cacoullos1966estimation}:
\begin{align*}
    p_{\text{KDE}}(x,t) & = \frac{1}{Nh^{n}}\sum_{i=1}^{N}K\Big(\frac{1}{h}(x-X_t^i)\Big),
\end{align*}
where $K(x)$ is a kernel function and $h$ is the bandwidth.
% The Gaussian kernel is frequently used:
% $$
% K(x)=\frac{1}{(2\pi)^{n/2}}\exp\big(-\frac{1}{2}x^Tx\big).
% $$
We treat $p_{\text{KDE}}(x,t)$ as a noisy measurement of $p(x,t)$. 
% (See Appendices for how to construct $p_{\text{KDE}}$.)
We also treat $\nabla p_{\text{KDE}}(x,t)$, the gradient of $p_{\text{KDE}}(x,t)$, as a noisy measurement of $q(x,t)$.
Define $w(x,t)=p_{\text{KDE}}(x,t)-p(x,t)$ and $w_g(x,t)=\nabla p_{\text{KDE}}(x,t)-q(x,t)$.
Then, $w(x,t)$ and $w_g(x,t)$ are approximately infinite-dimensional Gaussian noises with diagonal covariance operators $\bar{\mathcal{R}}(t)=k\operatorname{diag}(p(t))$ and $\bar{\mathcal{R}}_g(t)=k_g\operatorname{diag}(q(t))$, respectively, where $k,k_g>0$ are constants depending on the kernels and $N$.
(See Appendices in \cite{zheng2021distributedmean} for why $w$ and $w_g$ are approximately Gaussian.)

We obtain two infinite-dimensional linear systems:
\begin{align}
&\left\{\begin{array}{l}
\dot{p}=\mathcal{A}(t) p \\
y=p_{\text{KDE}}(t)=p+w(t)
\end{array}\right. \label{eq:density evolution equation}\\
&\left\{\begin{array}{l}
\dot{q}=\mathcal{A}_{g}(t) q \\
y_{g}=\nabla p_{\text{KDE}}(t)=q+w_{g}(t)
\end{array}\right. \label{eq:gradient evolution equation}
\end{align}
We can design infinite-dimensional Kalman filters for \eqref{eq:density evolution equation} and \eqref{eq:gradient evolution equation}.
% \begin{align}
% \begin{split}
%     &\dot{\hat{p}}=\mathcal{A}(t)p+\bar{\mathcal{P}}\bar{\mathcal{R}}^{-1}(t)(y-\hat{p}),\quad p(t_0)=p_{\text{KDE}}(t_0)\\
%     &\dot{\bar{\mathcal{P}}}=\mathcal{A}(t)\bar{\mathcal{P}}+\bar{\mathcal{P}}\mathcal{A}^{*}(t)-\bar{\mathcal{P}}\bar{\mathcal{R}}^{-1}(t)\bar{\mathcal{P}},\quad \bar{\mathcal{P}}(t_0)=\bar{\mathcal{P}}_0
% \end{split}
% \end{align}
% and 
% \begin{align}
% \begin{split}
%     &\dot{\hat{q}}=\mathcal{A}_g(t)g+\bar{\mathcal{P}}_g\bar{\mathcal{R}}_g^{-1}(t)(y_g-\hat{q}),\quad g(t_0)=\nabla p_{\text{KDE}}(t_0)\\
%     &\dot{\bar{\mathcal{P}}}_g=\mathcal{A}_g(t)\bar{\mathcal{P}}_g+\bar{\mathcal{P}}_g\mathcal{A}_g^{*}(t)-\bar{\mathcal{P}}_g\bar{\mathcal{R}}_g^{-1}(t)\bar{\mathcal{P}}_g,\quad \bar{\mathcal{P}}_g(t_0)=\bar{\mathcal{P}}_{g0}.
% \end{split}
% \end{align}
However, $\bar{\mathcal{R}}$ and $\bar{\mathcal{R}}_g$ are unknown because they depend on $p$ and $q$, the states to be estimated.
Hence, we approximate $\bar{\mathcal{R}}$ and $\bar{\mathcal{R}}_g$ with $\mathcal{R}(t)=k\operatorname{diag}(p_{\text{KDE}}(t))$ and $\mathcal{R}_g(t)=k_g\operatorname{diag}(\nabla p_{\text{KDE}}(t))$.
In this way, the ``suboptimal'' density filter is given by:
\begin{align}
    &\dot{\hat{p}}=\mathcal{A}(t)\hat{p}+\mathcal{P}(t)\mathcal{R}^{-1}(t)(y(t)-\hat{p}), \label{eq:density filter}\\
    &\dot{\mathcal{P}}=\mathcal{A}(t)\mathcal{P}+\mathcal{P}\mathcal{A}^{*}(t)-\mathcal{P}\mathcal{R}^{-1}(t)\mathcal{P}, \label{eq:density Riccati}
\end{align}
% with $\hat{p}(0)=p_{\text{KDE}}(0)$ and $\mathcal{P}(0)=\mathcal{P}_0$, 
and the ``suboptimal'' gradient filter is given by:
\begin{align}
    &\dot{\hat{q}}=\mathcal{A}_g(t)\hat{q}+\mathcal{Q}(t)\mathcal{R}_g^{-1}(t)(y_g(t)-\hat{q}), \label{eq:gradient filter}\\
    &\dot{\mathcal{Q}}=\mathcal{A}_g(t)\mathcal{Q}+\mathcal{Q}\mathcal{A}_g^{*}(t)-\mathcal{Q}\mathcal{R}_g^{-1}(t)\mathcal{Q}, \label{eq:gradient Riccati}
\end{align}
% with $\hat{q}(0)=\nabla p_{\text{KDE}}(0)$ and $\mathcal{Q}(0)=\mathcal{Q}_{0}$, 
where $\hat{p}$ and $\hat{q}$ are our estimates for $p$ and $q$ (i.e., $\nabla p$).
 
Now we study the stability and optimality of these two filters.
Let $\bar{\mathcal{P}}$ and $\bar{\mathcal{Q}}$ be the corresponding solutions of \eqref{eq:density Riccati} and \eqref{eq:gradient Riccati} when $\mathcal{R}$ and $\mathcal{R}_g$ are respectively replaced by $\bar{\mathcal{R}}$ and $\bar{\mathcal{R}}_g$, their true but unknown values.
Then, $\bar{\mathcal{P}}$ and $\bar{\mathcal{Q}}$ represent the optimal flows of estimation error covariance.
We denote $\mathcal{L}=\mathcal{P}\mathcal{R}^{-1}$ and $\mathcal{L}_g=\mathcal{Q}\mathcal{R}_g^{-1}$ to represent the suboptimal Kalman gains, and denote $\bar{\mathcal{L}}=\bar{\mathcal{P}}\bar{\mathcal{R}}^{-1}$ and $\bar{\mathcal{L}}_g=\bar{\mathcal{Q}}\bar{\mathcal{R}}_g^{-1}$ to represent the optimal Kalman gains.
Define $\epsilon=\hat{p}-p$ and $\epsilon_g=\hat{q}-q$ as the estimation errors of the filters.
Then $\epsilon$ and $\epsilon_g$ satisfy the follow equations:
\begin{align} 
    &\Dot{\epsilon} = (\mathcal{A}(t)-\mathcal{P}(t)\mathcal{R}^{-1}(t))\epsilon+\mathcal{P}(t)\mathcal{R}^{-1}(t)w \label{eq:mean field estimation error equation}\\
    &\Dot{\epsilon}_g = (\mathcal{A}_g(t)-\mathcal{Q}(t)\mathcal{R}_g^{-1}(t))\epsilon_g+\mathcal{Q}(t)\mathcal{R}_g^{-1}(t)w_g\label{eq:gradient estimation error equation}.
\end{align}
We can show that the suboptimal filters are stable and remain close to the optimal ones.
The following theorem is for the density filter \eqref{eq:density filter} and \eqref{eq:density Riccati}, which is proved in \cite{zheng2020pde,zheng2021distributedmean}.

\begin{theorem} \label{thm:stability of density filter}
Assume that $\|\mathcal{P}(t)\|$ and $\|\Bar{\mathcal{P}}(t)\|$ are uniformly bounded, and $\exists c_1,c_2>0$ such that for $t\geq 0$,
\begin{equation}
    c_1\mathcal{I}\leq \mathcal{R}^{-1}(t),\Bar{\mathcal{R}}^{-1}(t),\mathcal{P}^{-1}(t),\Bar{\mathcal{P}}^{-1}(t)\leq c_2\mathcal{I}.
\end{equation}
Then we have:
\begin{itemize}
    \item[(i)] $\|\epsilon\|_{L^2}$ is ISS with respect to (w.r.t.) $\|w\|_{L^2}$ (and is uniformly exponentially stable if $w=0$);
    \item[(ii)] $\|\bar{\mathcal{P}}-\mathcal{P}\|$ is LISS w.r.t. $\|\Bar{\mathcal{R}}^{-1}-\mathcal{R}^{-1}\|$;
    \item[(iii)] $\|\Bar{\mathcal{L}}-\mathcal{L}\|$ is LISS w.r.t. $\|\Bar{\mathcal{R}}^{-1}-\mathcal{R}^{-1}\|$.
\end{itemize}
\end{theorem}

Property (i) means that the suboptimal filter \eqref{eq:density filter} is stable even if $\mathcal{R}$ is only an approximation for $\bar{\mathcal{R}}$.
Property (ii) means that the solution of the suboptimal operator Riccati equation \eqref{eq:density Riccati} remains close to the solution of the optimal one (when $\mathcal{R}$ is replaced by $\bar{\mathcal{R}}$).
Property (iii) means that the suboptimal Kalman gain remains close to the optimal Kalman gain.
We have similar results for the gradient filter \eqref{eq:gradient filter} and \eqref{eq:gradient Riccati}.
The proof is similar to the proof of Theorem \ref{thm:stability of density filter} and thus omitted.

\begin{theorem} \label{thm:stability of gradient filter}
Assume that $\|\mathcal{Q}(t)\|$ and $\|\Bar{\mathcal{Q}}(t)\|$ are uniformly bounded, and $\exists c_3,c_4>0$ such that for $t\geq 0$,
\begin{equation}
    c_3\mathcal{I}\leq \mathcal{R}_g^{-1}(t),\Bar{\mathcal{R}}_g^{-1}(t),\mathcal{Q}^{-1}(t),\Bar{\mathcal{Q}}^{-1}(t)\leq c_4\mathcal{I}.
\end{equation}
Then we have:
\begin{itemize}
    \item[(i)] $\|\epsilon_g\|_{L^2}$ is ISS w.r.t. $\|w_g\|_{L^2}$ (and is uniformly exponentially stable if $w_g=0$);
    \item[(ii)] $\|\bar{\mathcal{Q}}-\mathcal{Q}\|$ is LISS w.r.t. $\|\Bar{\mathcal{R}}_g^{-1}-\mathcal{R}_g^{-1}\|$;
    \item[(iii)] $\|\Bar{\mathcal{L}}_g-\mathcal{L}_g\|$ is LISS w.r.t. $\|\Bar{\mathcal{R}}_g^{-1}-\mathcal{R}_g^{-1}\|$.
\end{itemize}
\end{theorem}

\subsection{Stability of feedback interconnection}

Now we discuss the stability of the feedback interconnection of density estimation and control.
We collect equations \eqref{eq:tracking error equation}, \eqref{eq:mean field estimation error equation} and \eqref{eq:gradient estimation error equation} in the following for clarity:
\begin{align} \label{eq:interconnected system}
\begin{split}
    &\partial_t\Phi=\nabla\cdot\left(\alpha p\frac{p_*(\epsilon_g+\nabla p)-(\epsilon+p)\nabla p_*}{p_*^2}+\sigma p_*\nabla\frac{p}{p_*}\right), \\
    &\Dot{\epsilon} = [\mathcal{A}(t;\Phi,\epsilon,\epsilon_g)-\mathcal{P}(t)\mathcal{R}^{-1}(t)]\epsilon+\mathcal{P}(t)\mathcal{R}^{-1}(t)w,\\
    &\Dot{\epsilon}_g = [\mathcal{A}_g(t;\Phi,\epsilon,\epsilon_g)-\mathcal{Q}(t)\mathcal{R}_g^{-1}(t)]\epsilon_g+\mathcal{Q}(t)\mathcal{R}_g^{-1}(t)w_g,
\end{split}
\end{align}
where we write $\mathcal{A}(t;\Phi,\epsilon,\epsilon_g)$ and $\mathcal{A}_g(t;\Phi,\epsilon,\epsilon_g)$ to emphasize their dependence on $\Phi$, $\epsilon$ and $\epsilon_g$ through $v$.

A critical observation is that the ISS results we have established for $\epsilon$ and $\epsilon_g$ are valid in spite of this dependence.
This is because when designing estimation algorithms, $\mathcal{A}$ and $\mathcal{A}_g$ (and $v$) are treated as known system coefficients.
In this way, the bilinear control system \eqref{eq:FP equation} becomes a linear system.
The dependence on $\Phi$, $\epsilon$ and $\epsilon_g$ can be seen as part of the time-varying nature of $\mathcal{A}$ and $\mathcal{A}_g$.
Hence, the stability results for our density and gradient filters will not be affected.
(Nevertheless, since in the interconnection $v$ depends on $(\hat{p},\hat{q})$ and becomes stochastic, the presented filters downgrade to be the best ``linear'' filters instead of the minimum covariance filters.)
In this regard, we can treat \eqref{eq:interconnected system} as a cascade system.
By Theorem \ref{thm:(L)ISS cascade}, we have the following stability result for \eqref{eq:interconnected system}.

\begin{theorem}[Interconnected stability]
Under the assumptions in Theorems \ref{thm:stability of density filter} and \ref{thm:stability of gradient filter}, $\|\Phi\|_{L^2}$, $\|\epsilon\|_{L^2}$ and $\|\epsilon_g\|_{L^2}$ are all ISS w.r.t. $\|w\|_{L^2}$ and $\|w_g\|_{L^2}$.
\end{theorem}

% \begin{proof}
% The ISS properties of $\|\epsilon\|_{L^2}$ and $\|\epsilon_g\|_{L^2}$ have been proved before, and the ISS property of $\|\Phi\|_{L^2}$ follows from the argument that the cascade connection of two ISS systems is still ISS.
% \end{proof}

This theorem has two implications.
First, the stability results for the filters are independent of the density controller.
Hence, they can be used for any control design when density feedback is required.
Second, the interconnected system is always ISS as long as the density feedback laws are designed such that the tracking error is ISS w.r.t. the $L^2$ norms of estimation errors.
Considering that norms of infinite-dimensional vectors are not equivalent in general, obtaining ISS results w.r.t. $L^2$ norms is critical.
This in turn highlights the advantage of \eqref{eq:modified density feedback law 1}, because it is difficult to obtain such an ISS result if $p$ presents in the denominator, as in \cite{eren2017velocity, elamvazhuthi2018bilinear, zheng2021transporting}.

\begin{remark}
In our recent work \cite{zheng2021distributedmean}, we proved that with some regularity assumptions on $v$, the density filter \eqref{eq:density filter} alone also produces convergent gradient estimates, which means the gradient filter may not be needed.
However, those assumptions are not satisfied in a feedback interconnected system considered in this work.
Hence, the gradient filter \eqref{eq:gradient filter} still serves as a general solution to estimate the gradient.
Determining the condition such that the density filter \eqref{eq:density filter} produces convergent gradient estimates in feedback interconnected systems is the subject of continuing research.
% The density control strategy in this work is essentially centralized because of the requirement of knowing the real-time density (and its gradient).
% This constraint can be relaxed if each robot is able to estimate the density in a distributed way.
% In \cite{zheng2021distributedmean}, we have presented a distributed density filter for this purpose, where each robot estimates the global density using only local observation and communication.
% The interconnected stability of distributed density estimation and control is left as our future work.
\end{remark}

\section{Simulation studies}
\label{section:simulation}

\begin{figure*}[t]
\setlength{\abovecaptionskip}{0.0cm}
\setlength{\belowcaptionskip}{-0.5cm}
    \centering
    \begin{subfigure}[b]{0.20\textwidth}
        \centering
        \includegraphics[width=\textwidth]{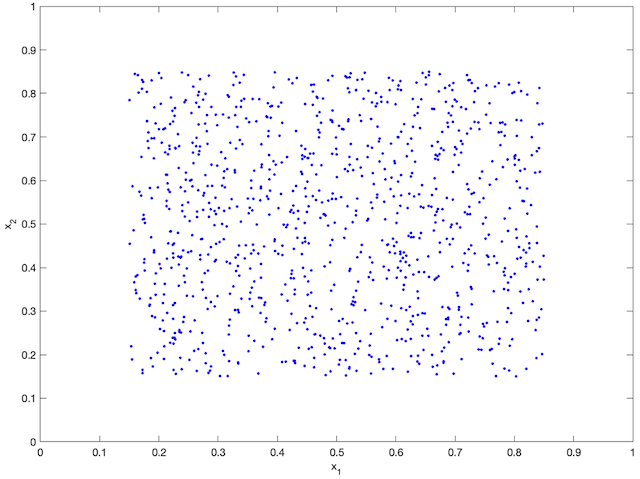}
    \end{subfigure}
    \begin{subfigure}[b]{0.20\textwidth}
        \centering
        \includegraphics[width=\textwidth]{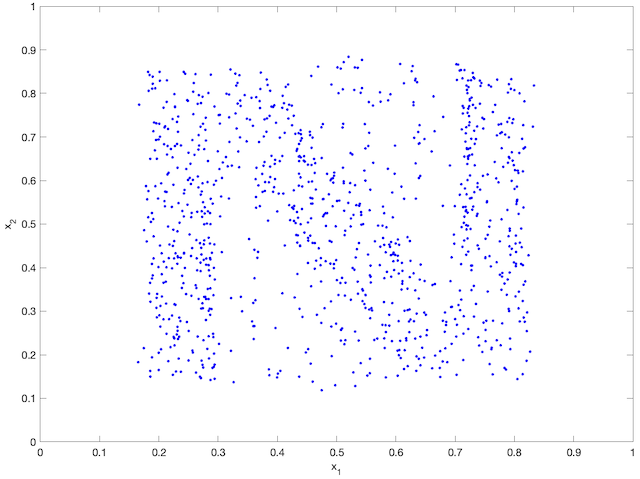}
    \end{subfigure}
    \begin{subfigure}[b]{0.20\textwidth}
        \centering
        \includegraphics[width=\textwidth]{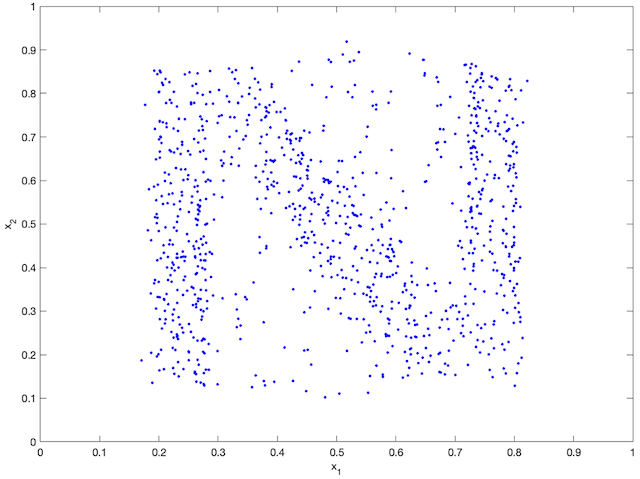}
    \end{subfigure}
    \begin{subfigure}[b]{0.20\textwidth}
        \centering
        \includegraphics[width=\textwidth]{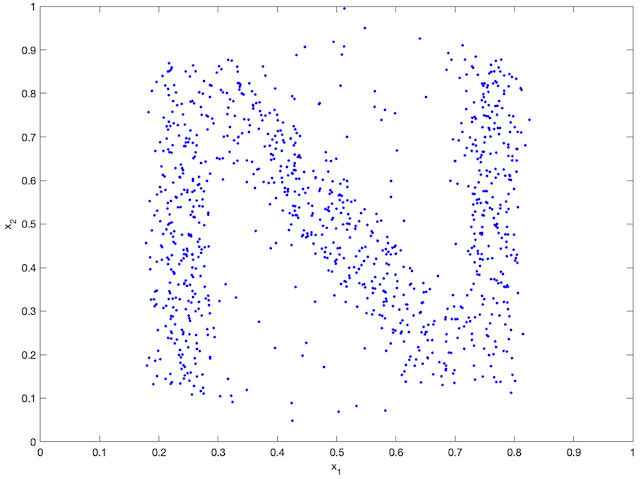}
    \end{subfigure}
    
    \begin{subfigure}[b]{0.20\textwidth}
        \centering
        \includegraphics[width=\textwidth]{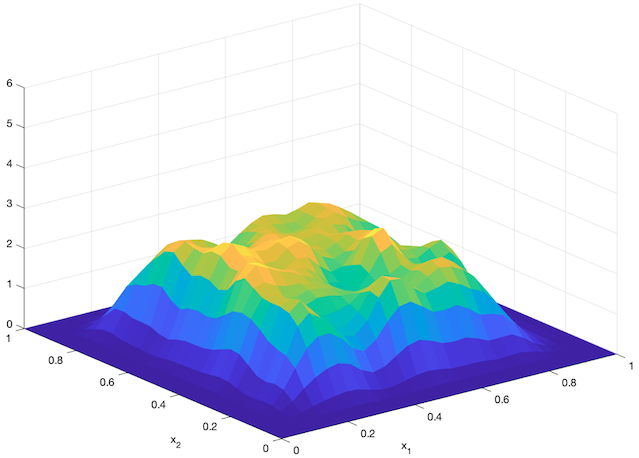}
    \end{subfigure}
    \begin{subfigure}[b]{0.20\textwidth}
        \centering
        \includegraphics[width=\textwidth]{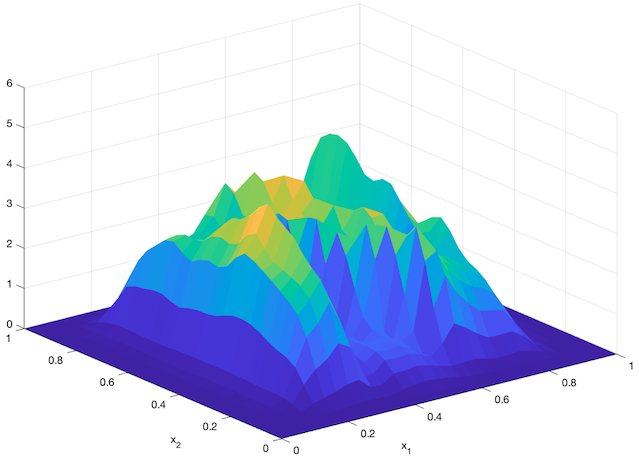}
    \end{subfigure}
    \begin{subfigure}[b]{0.20\textwidth}
        \centering
        \includegraphics[width=\textwidth]{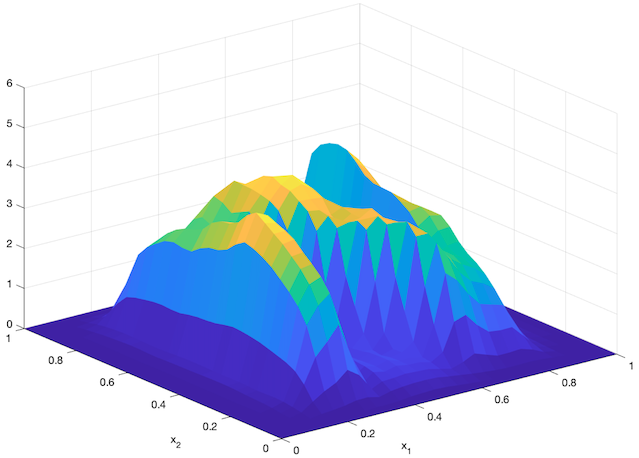}
    \end{subfigure}
    \begin{subfigure}[b]{0.20\textwidth}
        \centering
        \includegraphics[width=\textwidth]{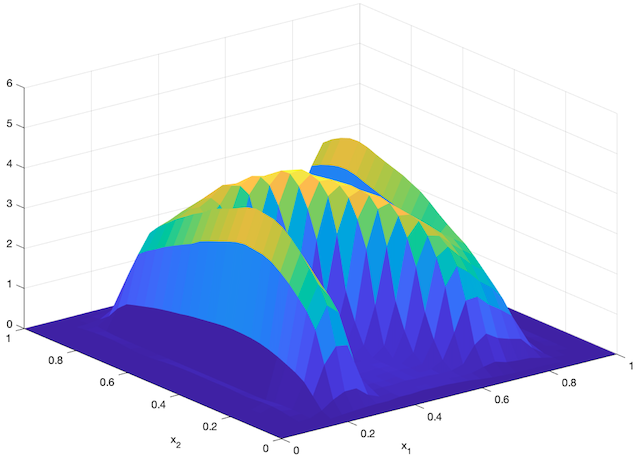}
    \end{subfigure}
    
    \begin{subfigure}[b]{0.20\textwidth}
        \centering
        \includegraphics[width=\textwidth]{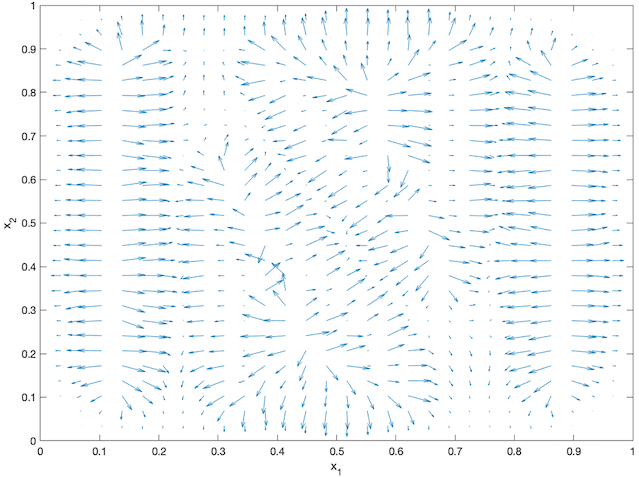}
    \end{subfigure}
    \begin{subfigure}[b]{0.20\textwidth}
        \centering
        \includegraphics[width=\textwidth]{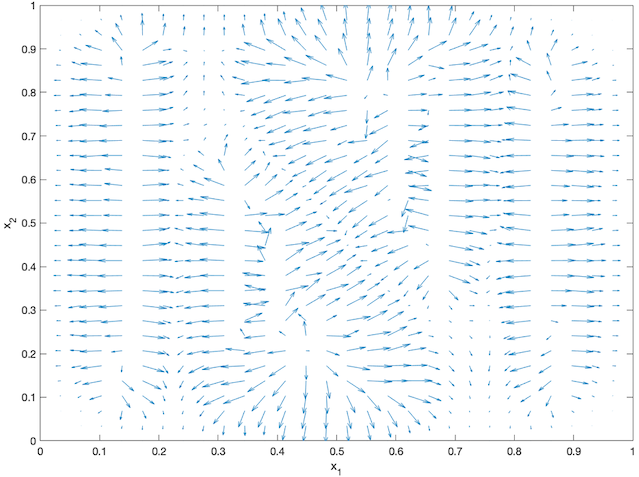}
    \end{subfigure}
    \begin{subfigure}[b]{0.20\textwidth}
        \centering
        \includegraphics[width=\textwidth]{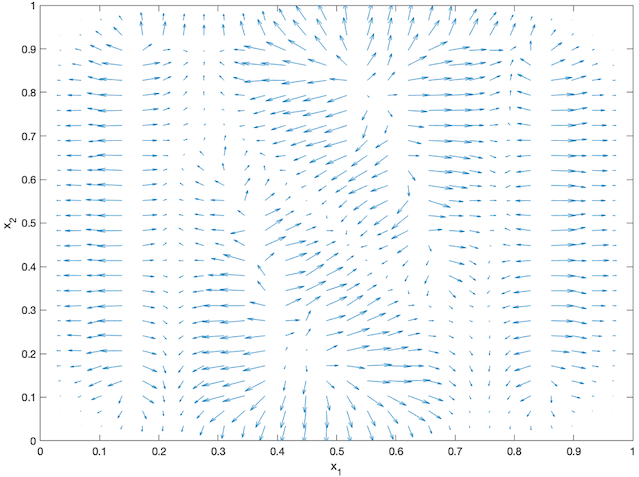}
    \end{subfigure}
    \begin{subfigure}[b]{0.20\textwidth}
        \centering
        \includegraphics[width=\textwidth]{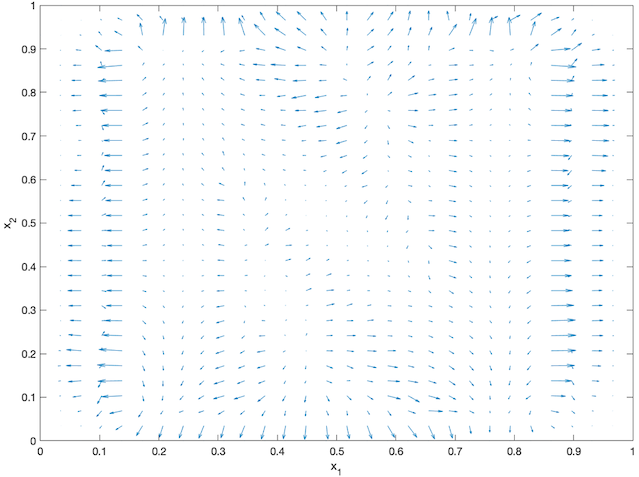}
    \end{subfigure}
    \caption{Evolution of the swarm (top), the density estimates $\Hat{p}(x,t)$ (middle) and the velocity fields ${v}(x,t)$ (bottom). Each column represents a single time step.}
    \label{fig:pdf evolution}
\end{figure*}

An agent-based simulation using 1024 robots is performed on Matlab to verify the proposed control law. 
We set $\Omega=(0,1)^2$, $\sqrt{2\sigma}=0.01$ and $\alpha=0.003$.
Each robot is simulated according to \eqref{eq:Langevin equation} where $v$ is given by \eqref{eq:modified density feedback law using estimates}. 
The initial positions are drawn from a uniform distribution on $[0.15,0.85]^2$. 
The desired density $p_*(x)$ is illustrated in Fig. \ref{fig:desired density}.
The implementation of the filters and the feedback controller is based on the finite difference method.
% That is, the densities and the operators are approximated by finite-dimensional vectors and matrices.
We discretize $\Omega$ into a $30\times30$ grid, and the time difference is $0.01s$. 
We use KDE (in which we set $h=0.04$) to obtain $p_{\text{KDE}}$ and $\nabla p_{\text{KDE}}$.
Numerical implementation of the filters are introduced in \cite{zheng2021distributedmean}.
Simulation results are given in Fig. \ref{fig:pdf evolution}.
It is seen that the swarm is able to evolve towards the desired density. 
The convergence error $\|\Hat{p}-p^*\|_{L^2(\Omega)}$ is given in Fig. \ref{fig:convergence error}, which shows that the error converges exponentially to a small neighbourhood around $0$ and remains bounded, which verifies the ISS property of the proposed algorithm. 

\begin{figure}[hbt!]
\setlength{\abovecaptionskip}{0.2cm}
\setlength{\belowcaptionskip}{-0.0cm}
    \centering
    \begin{subfigure}[b]{0.22\textwidth}
        \centering
        \includegraphics[width=\textwidth]{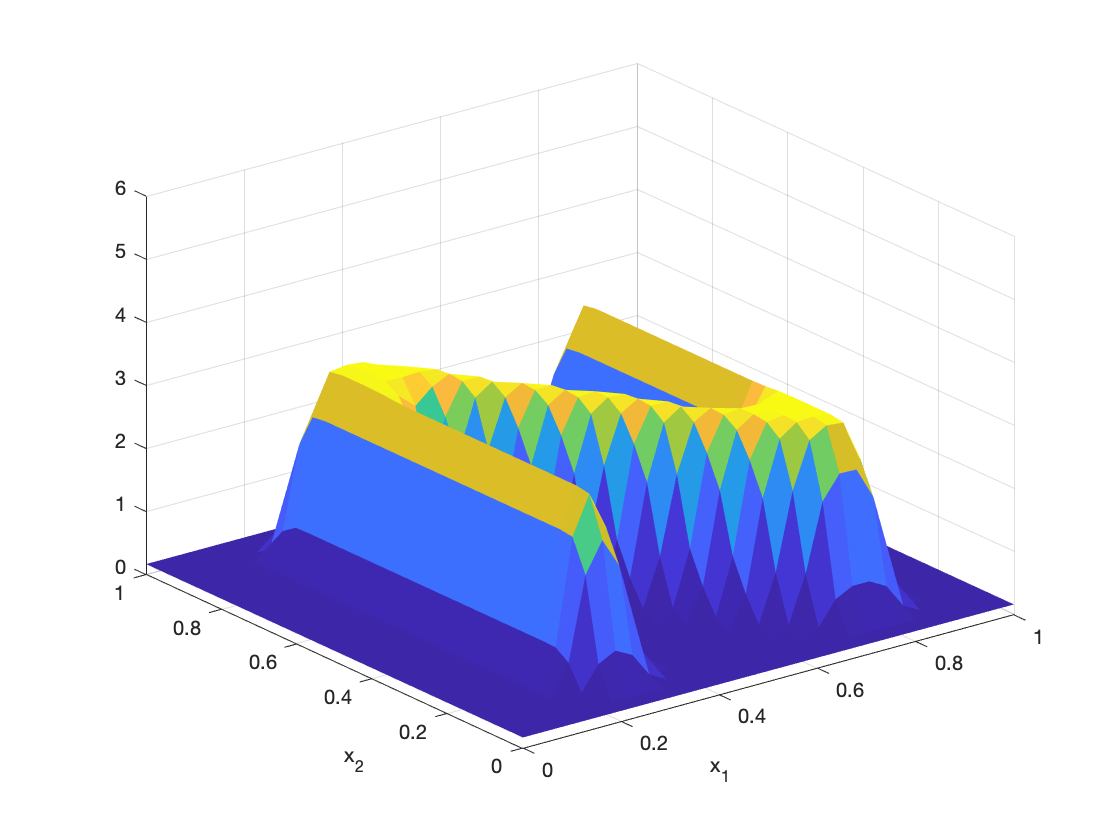}
        \caption{Target density $p_*(x)$.}
        \label{fig:desired density}
    \end{subfigure}
    \begin{subfigure}[b]{0.22\textwidth}
        \centering
        \includegraphics[width=\textwidth]{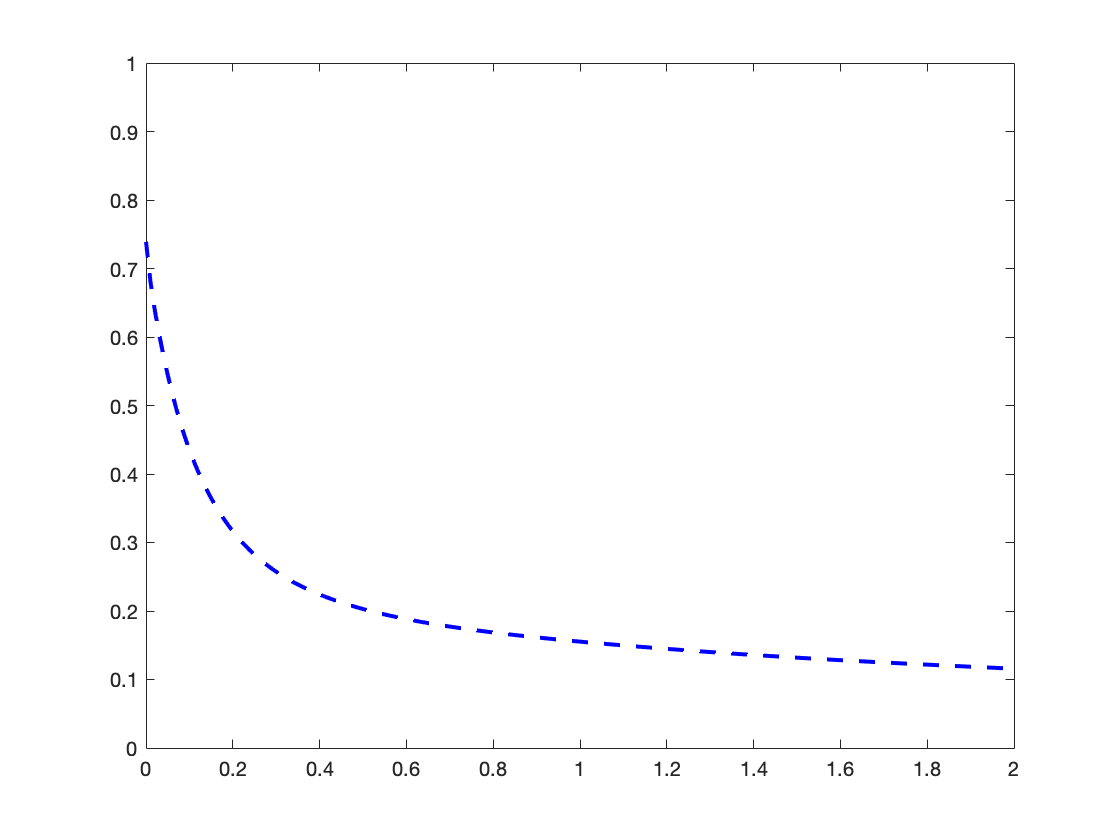}
        \caption{Convergence error.}
        \label{fig:convergence error}
    \end{subfigure}
    \caption{}
\end{figure}

\section{Conclusion}
This paper studied the interplay of density estimation and control algorithms.
We proposed new density feedback laws for robust density control of swarm robotic systems and filtering algorithms for estimating the mean-field density and its gradient.
We also proved that the interconnection of these algorithms is globally ISS.
% In implementation, we need a communication center to perform the estimation algorithms and broadcast the estimates to the robots.
% This is suitable for surveillance applications such as UAV-based environment surveillance.
Our future work is to incorporate the density control algorithm for higher-order nonlinear systems in \cite{zheng2021backstepping} and the distributed density estimation algorithm in \cite{zheng2021distributedmean}, and study their interconnected stability.

\bibliographystyle{IEEEtran}
\bibliography{References}

\end{document}